\documentclass[runningheads, 11pt]{llncs}
\usepackage{amsfonts}
\usepackage{cite}
\usepackage{amsmath}
\usepackage{cases}
\usepackage{color}
\usepackage{amssymb}
\usepackage{graphicx}
\usepackage{lineno}
\newcommand{\tr}{{\rm Tr}}

\def\fqu#1 {\fbox {\footnote {\ }}\ \footnotetext { From Qu: {\color{red}#1}}}
\def\fxh#1 {\fbox {\footnote {\ }}\ \footnotetext { From Hai: {\color{blue}#1}}}

\usepackage{url}
\urldef{\mailsa}\path|xiong.hai@163.com, ljqu_happy@hotmail.com|

\begin{document}
\mainmatter  

\title{A Note on Cross Correlation Distribution of Ternary $m$-Sequences}

\titlerunning{A Note on Cross Correlation Distribution of Ternary $m$-Sequences}

%

\author{Hai Xiong\and Longjiang Qu}
\authorrunning{Xiong Hai, Longjiang Qu}

\institute{College of Science,
National University of Defense Technology, Changsha  410073, China
\mailsa\\
}

\toctitle{Lecture Notes in Computer Science}
\tocauthor{Authors' Instructions}
\maketitle

\begin{abstract}
In this note, we prove a conjecture proposed by Tao Zhang, Shuxing Li, Tao Feng and Gennian Ge,
IEEE Transaction on Information Theory, vol. 60, no. 5, May 2014.
This conjecture is about the cross correlation distribution of ternary $m$-sequences.

\end{abstract}

\section{Introduction}
Let $\omega$ be a $3$-rd complex root of unity.
Let $\{a_t\}$ and $\{b_t\}$ be two ternary sequences of period $N$.
The cross correlation function of $\{a_t\}$ and $\{b_t\}$ is defined by
$$C_{a,b}(\tau)=\sum_{t=0}^{N-1}\omega^{a_{t+\tau}-b_t},\ \tau\in\mathbb{Z}/(N).$$

Let $n$ be a positive integer. Let $\chi$ be an additive character of $GF(3^n)$ which
is defined by $\chi(x)=\omega^{\tr^n_1(x)}$, where $\tr^n_1(x)=x+x^3+\cdots+x^{3^{n-1}}$
is the trace function from $GF(3^n)$ to $GF(3)$. Generally, for a positive integer $r|n$,
 we denote  by $\tr^n_r(x)$  the trace function from $GF(3^n)$ to $GF(3^r)$, which is defined
 by $\tr^n_r(x)=x+x^{3^r}+\cdots+x^{3^{r(\frac{n}{r}-1)}}$.

Let $\{a_t\}$ be a ternary $m$-sequence of period $3^n-1$ and {let} $\{b_t\}$ be its   $d$-decimation
 where $\gcd(d, 3^n-1)=1$.
 Let $\alpha$ be a primitive element of $GF(3^n)$.
 Then we denote $C_{a, b}(\tau)$  by $C_d(z)$,
 where $z=\alpha^\tau$. Clearly, we have $C_d(z)+1=\sum_{x\in GF(3^n)}\chi(zx-x^d)$. And we define $S_d(z):=\sum_{x\in GF(3^n)}\chi(zx-x^d)$.
Hence computing the cross correlation distribution of $m$-sequences  is equivalent
to compute the values distribution of Weil sum $S_d(z)$.

The cross correlation distribution of $m$-sequences is an important topic in
sequences, coding theory and communications.
It essentially arises
in many contexts with various names, please refer to the
appendix of \cite{Katz:2012} for more details.
Many results on this topic have been reported.
Please see \cite{Zhang:2014} for an exhaustive survey.
In \cite{Zhang:2014}, Zhang et al. also determined the distribution of cross correlation of some
ternary $m$-sequences and got some interesting  results  about the cross correlation of some binary
$m$-sequences. One of their main results can be presented as follows.
\begin{theorem}\cite[Theorem II.5]{Zhang:2014}\label{th1}
Let $r$ be a positive integer such that $\gcd(r,3)=1$.
Let $n=3r$, $d=3^r+2$ or $d=3^{2r}+2$. Let $s$ be a ternary
$m$-sequence of period $3^n-1$. Then the cross correlation values between $s$
and its $d$-decimation are showed as in Tables \ref{tab1} and \ref{tab2}.
\end{theorem}

\begin{table}[]
\tabcolsep 0pt
\caption{The distribution of the case $r$ even}\label{tab1}
\begin{center}
\def\temptablewidth{0.5\textwidth}
{\rule{\temptablewidth}{1pt}}
\begin{tabular*}{\temptablewidth}{@{\extracolsep{\fill}}cc}
Cross Correlation Value  & Occurs Times \\   \hline
       $-1$& $\frac{3^{3r}+3^{2r}}{2}-3^r$\\
       $3^{2r}-1$& $3^r$\\
       $3^{\frac{3r}{2}}-1$& $\frac{3^{3r-1}+3^{2r-1}}{2}$\\
       $-3^{\frac{3r}{2}}-1$& $\frac{3^{3r-1}+3^{2r-1}}{2}$\\
       $2\cdot3^{\frac{3r}{2}}-1$& $\frac{3^{3r-1}+3^{2r-1}}{4}$\\
       $-2\cdot3^{\frac{3r}{2}}-1$& $\frac{3^{3r-1}+3^{2r-1}}{4}$
       \end{tabular*}
       {\rule{\temptablewidth}{1pt}}
       \end{center}
       \end{table}

 \begin{table}[]
\tabcolsep 0pt
\caption{The distribution of the case $r$ odd}\label{tab2}
\begin{center}
\def\temptablewidth{0.5\textwidth}
{\rule{\temptablewidth}{1pt}}
\begin{tabular*}{\temptablewidth}{@{\extracolsep{\fill}}cc}
Cross Correlation Value  & Occurs Times \\   \hline
       $-1$& $2\cdot3^{3r-1}+3^{2r-1}-3^r$\\
       $3^{2r}-1$& $3^r$\\
       $3^{\frac{3r+1}{2}}-1$& $\frac{3^{3r-1}+3^{2r-1}}{2}$\\
       $-3^{\frac{3r+1}{2}}-1$& $\frac{3^{3r-1}+3^{2r-1}}{2}$
       \end{tabular*}
       {\rule{\temptablewidth}{1pt}}
       \end{center}
       \end{table}

Zhang et al.  conjectured that Theorem \ref{th1}  is also right if $\gcd(r, 3)=3$.
In this paper we prove their conjecture for any positive integer $r$.
 And our technique is generalized from theirs.

In the rest of this paper, we always assume that $r$ is a positive {integer, $d=3^r+2\text{ or }3^{2r}+2$,}
{and} $n=3r$. Let $E=GF(3^n)$ and $F=GF(3^r)$. {It is easy to verify that $\gcd(d,3^n-1)=1$.}

\section{A Proof}

Before proving the conjecture, let us review the sketch of the proof of Theorem 1 in \cite{Zhang:2014}.
%
At first, a suitable element was chosen to construct field extension from $F$ to $E$.
Hence {e}very element in $E$ can be expressed by
some elements in the subfield $F$ {with the aforementioned element.  And  then  $S_d(z)$ can be expressed by some exponential sums over $F$.} Finally,
$S_d(z)$ was computed by using {some characterizations of} quadratic Weil sum {\cite[Lemma II.2]{Zhang:2014} and}
quadratic Gauss sum  {\cite[Lemma II.1]{Zhang:2014}.}

In the following, we will prove the conjecture. The main difference between our proof and the one in \cite{Zhang:2014} is that
we choose different elements to construct field extension.
Then we can remove the restriction $\gcd(r, 3)=1$ and the discussion {of cases}  $r\equiv 2, 1\pmod 3$ in the proof of \cite{Zhang:2014}.

The rest of the paper is split into two cases according to the value of $d$.

\textbf{Case $d=3^r+2$:}

Let $u$ be an element of $F$ such
that $\tr^r_1(u-1)=1$. Hence  $x^3-x-(u-1)^3$ is an irreducible polynomial over $F$.
And let $\alpha$ be a root of $x^3-x-(u-1)^3=0$.
 Then we can get $E=F(\alpha)$,
which means that for any $x\in E$, it can be uniquely expressed as $x=x_0+x_1\alpha+x_2\alpha^2$,
where $x_0, x_1, x_2\in F$.

\begin{lemma}\label{le5}
Let $x=x_0+x_1\alpha+x_2\alpha^2$ and $z=z_0+z_1\alpha+z_2\alpha^2$ be two elements of $E$.
Then we have
$$\tr^n_r(x^d)={((u-1)^3+1)x_2^3+x_2^2x_1+x_2^2x_0+2x_2x_1^2+2x_1^3}$$
and
$$\tr^n_r(zx)={2(z_2+z_0)x_2+2z_1x_1+2z_2x_0}.$$
 \end{lemma}
\begin{proof}
Noting that $\alpha^3=\alpha+(u-1)^3$, we can deduce that $\alpha^{3^r}=\alpha+\tr^r_1((u-1)^3)=\alpha+1$.
Thus $x^{3^r}=x_0+x_1(\alpha+1)+x_2(\alpha+1)^2$ and
$x^{3^{2r}}=x_0+x_1(\alpha+2)+x_2(\alpha+2)^2$.
Hence we have
\begin{equation}
\begin{split}
\tr^n_r(x^d)
=&\quad x^d+(x^d)^{3^{r}}+(x^d)^{3^{2r}}\\
=&\quad x^{3^r}x^2+x^{3^{2r}}(x^{3^r})^2+x(x^{3^{2r}})^2\\
=&\quad (x_0+x_1(\alpha+1)+x_2(\alpha+1)^2)(x_0+x_1\alpha+x_2\alpha^2)^2\\
&+(x_0+x_1(\alpha+2)+x_2(\alpha+2)^2)(x_0+x_1(\alpha+1)+x_2(\alpha+1)^2)^2\\
&+(x_0+x_1\alpha+x_2\alpha^2)(x_0+x_1(\alpha+2)+x_2(\alpha+2)^2)^2\\
=&\quad {((u-1)^3+1)x_2^3+x_2^2x_1+x_2^2x_0+2x_2x_1^2+2x_1^3}.
\end{split}
\nonumber
\end{equation}
The last step is got from a complicated but not difficult computation.
Comparing with the above equation, it is easier to verify that
\begin{equation}
\begin{split}
\tr^n_r(zx)
=&\quad zx+(zx)^{3^r}+(zx)^{3^{2r}}\\
=&\quad (z_0+z_1\alpha+z_2\alpha^2)(x_0+x_1\alpha+x_2\alpha^2)\\
&+(z_0+z_1(\alpha+1)+z_2(\alpha+1)^2)(x_0+x_1(\alpha+1)+x_2(\alpha+1)^2)\\
&+(z_0+z_1(\alpha+2)+z_2(\alpha+2)^2)(x_0+x_1(\alpha+2)+x_2(\alpha+2)^2)\\
=&\quad{2(z_2+z_0)x_2+2z_1x_1+2z_2x_0}.
\nonumber
\end{split}
\end{equation}
$\hfill\Box$

According to Lemma \ref{le5}, we can get that
$$\tr^n_1(x^d)=\tr^r_1(\tr^n_r(x^d))=\tr^r_1({x_2^2x_1+x_2^2x_0+2x_2x_1^2+u x_2+2x_1})$$
and
$$\tr^n_1(zx)=\tr^r_1(\tr^n_r(zx))=\tr^r_1({2(z_2+z_0)x_2+2z_1x_1+2z_2x_0}).$$

Define $\chi_F(x)=\omega^{\tr^r_1(x)}$, for $x\in F$. Then we can get

\begin{equation}\label{eq2}
\begin{split}
&S_d(z)\\
=&\sum_{x\in E}\omega^{\tr^n_1(zx-x^d)}\\
=&\sum_{x_0, x_1, x_2\in F}\chi_F({x_2x_1^2+(2z_1-x_2^2+1)x_1+(2z_2+2z_0-u)x_2+(2z_2-x_2^2)x_0})\\
=&\sum_{x_1, x_2\in F}\chi_F({x_2x_1^2+(2z_1-x_2^2+1)x_1+(2z_2+2z_0-u)x_2})\sum_{x_0\in F}\chi_F((2z_2-x_2^2)x_0)\\
=&3^r\cdot \sum_{x_1\in F, x_2\in M}\chi_F({ x_2x_1^2+(2z_1-x_2^2+1)x_1+(2z_2+2z_0-u)x_2}),
\end{split}
\end{equation}
where $$M=\{x_2\in F| x_2^2=-z_2\}.$$

Following from similar arguments in \cite{Zhang:2014}, we can deduce the following results {by Eq. (\ref{eq2}),
\cite[Lemma II.1]{Zhang:2014} and \cite[Lemma II.2]{Zhang:2014}}. The details are omitted here.
\begin{itemize}
\item If $r$ is even, then $S_d(z)$  takes six values, namely $0$, $3^{2r}$, $3^{\frac{3r}{2}}$,
$-3^{\frac{3r}{2}}$, $2\cdot3^{\frac{3r}{2}}$ and $-2\cdot3^{\frac{3r}{2}}$.
And the number of  occurrences of  the first two values are $\frac{3^{3r}+3^{2r}}{2}-3^r$ and $3^r$ respectively.
\item If $r$ is odd, then $S_d(z)$  takes four values, namely $0$, $3^{2r}$, $3^{\frac{3r+1}{2}}$ and
$-3^{\frac{3r+1}{2}}$.
And the number of  occurrence of  the second value is  $3^r$.
\end{itemize}
Then according to \cite[Lemma II.3]{Zhang:2014} and \cite[Lemma II.4]{Zhang:2014}, we can solve the number of occurrences of  all the values.
The result is the same as Zhang et al.  conjectured in\cite{Zhang:2014}. Here a remark is as follows.
 In the original form of \cite[Lemma II.4]{Zhang:2014}, the authors assumed that $\gcd(r,3)=1$.
However, this result can be easily generalized to any positive integer $r$.

\textbf{Case $d=3^{2r}+2$:}

Let  $u$ be an element of {$F$ such that} $\tr^r_1(1-u)=1$ and let $\alpha$ be a root
of $x^3-x=(1-u)^3$. Then we also can get
$$\tr^n_1(x^d)=\tr^r_1(\tr^n_r(x^d))=\tr^r_1({2x_2^2x_1+x_2^2x_0+2x_2x_1^2+u x_2+x_1})$$
and
$$\tr^n_1(zx)=\tr^r_1(\tr^n_r(zx))=\tr^r_1({2(z_2+z_0)x_2+2z_1x_1+2z_2x_0}).$$

Then similarly as the first case, we can confirm the conjecture in this case.
\end{proof}

\section{Conclusion}
In this note, we completely determine the distribution of cross correlation values of a ternary
$m$-sequence with period $3^{3r}-1$ and its  $d$-decimation, where $d=3^{r}+2$ or $d=3^{2r}+2$.
Hence we confirm the conjecture presented in \cite{Zhang:2014}.

\section*{Acknowledgments}
The work of H. Xiong was supported by
 Hunan Provincial  Innovation Foundation for Postgraduate  (No. CX2013B007)
 and the Innovation Foundation of NUDT under Grant (No. B130201).
 The work of L. Qu was supported in part by
the Research Project of National University of Defense Technology under Grant CJ 13-02-01 and
the Program for New Century Excellent Talents in University (NCET).
 The first author also gratefully acknowledge financial support from China Scholarship Council.

\bibliographystyle{model1a-num-names}
\bibliography{<your-bib-database>}


\end{document}